\documentclass[conference]{IEEEtran}
\ifCLASSINFOpdf
\else
\fi
\usepackage[utf8]{inputenc}

\usepackage{framed, color}
\usepackage{ifthen}
\usepackage{multirow}
\usepackage{hyperref}
\usepackage{xspace}
\usepackage{relsize}  
\usepackage{alltt} 
\usepackage{listings}
\usepackage{pifont}
\usepackage[colorinlistoftodos,prependcaption,textsize=tiny]{todonotes}
\usepackage{xargs}
\usepackage{amsmath}
\usepackage{amsthm}
\usepackage{amssymb}
\usepackage{pifont}

\usepackage{tabularx}
\usepackage{rotating}
\usepackage{boxedminipage}
\usepackage{ascii}
\usepackage{multirow}

\usepackage[framemethod=tikz]{mdframed}
\newmdenv[innerlinewidth=0.5pt, roundcorner=4pt,linecolor=gray,innerleftmargin=0pt,
innerrightmargin=1pt,innertopmargin=2pt,innerbottommargin=0pt]{quotes}

\newif\ifdraft
\draftfalse

\newif\ifdoubleblind
\doubleblindfalse 

\usepackage{enumitem}
\setlist[enumerate]{leftmargin=1.5\parindent,labelindent=\parindent}
\setlist[itemize]{leftmargin=1.5\parindent,labelindent=\parindent}
\setlist[description]{leftmargin=\parindent}
\definecolor{DarkBlue}{rgb}{0.0859, 0.308, 0.523}
\definecolor{DarkOrange}{rgb}{0.8, 0.4, 0.0}
\definecolor{DarkGreen}{rgb}{0.00,0.40,0.00}
\definecolor{ScarletRed}{rgb}{0.60,0.00,0.00}
\definecolor{AlmostWhite}{rgb}{0.80,0.80,0.80}

\newenvironment{grammar}{\begin{alltt}\color{DarkBlue}\bfseries}{\end{alltt}\vspace{-.5\baselineskip}}
\newcommand{\lterm}[1]{\textcolor{DarkOrange}{\bfseries\ttfamily\MakeUppercase{#1}}}  
\let\rterm=\lterm  
\newcommand{\pterm}[1]{\rterm{\underline{#1}}}  
\newcommand{\expandsinto}{\textcolor{black}{::=}}   

\newenvironment{sample}{\begin{alltt}\color{DarkBlue}\bfseries}{\end{alltt}}
\newenvironment{sampleandarrow}{%
\begin{minipage}[b]{0.55\linewidth}
\begin{sample}
}{
\end{sample}
\end{minipage} 
\begin{minipage}[b]{2.1cm}
{\includegraphics*[width=2cm]{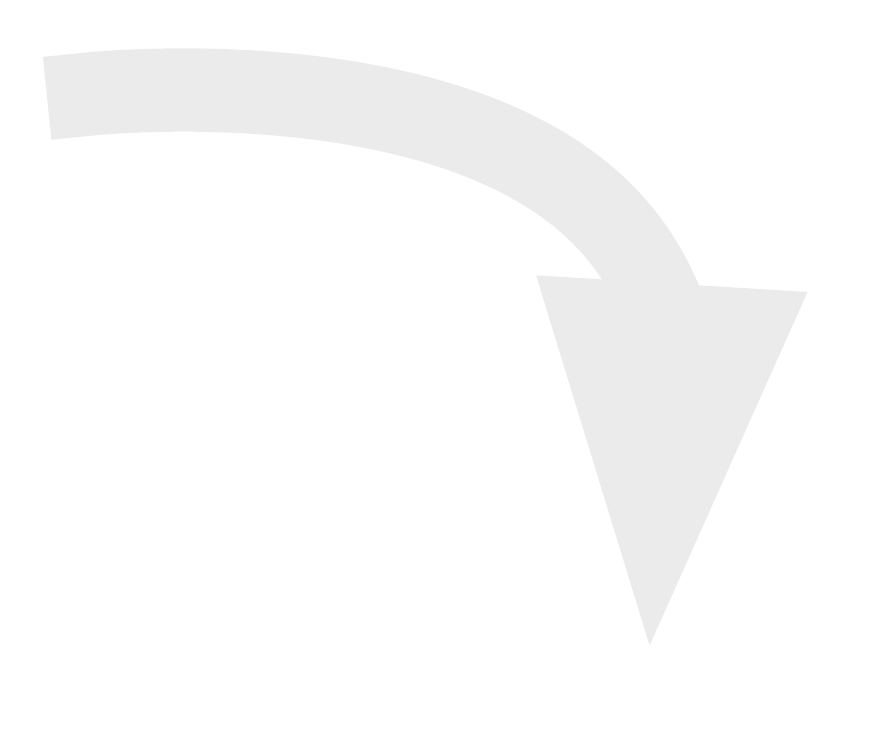}}
\end{minipage}
}

\newcommand{\mathid}[1]{\text{\rmfamily\textit{#1}}}

\def\|#1|{\mathid{#1}}

\newcommand{\codeid}[1]{\texttt{#1}}

\def\<#1>{\codeid{#1}}

\def\bs{\(\backslash\)}

\newcommand{\PASS}{\textcolor{DarkGreen}{\ding{52}}\xspace}
\newcommand{\FAIL}{\textcolor{ScarletRed}{\ding{56}}\xspace}

\newcommand{\reg}[1]{\ensuremath{\text{ \%#1}}}
\newcommand{\opcode}[4]{
& \\
$(f_0,s_0,\tau_0) \xrightarrow{#1} (f_1,_js_2,\tau_1)$: & \\
& \\
 $#2$ & $#3$ \\
 $#4$ & \\
 & \\
 \hline
}

\ifdoubleblind
\def\EXAGRAM{{\smaller EXAGRAM}\xspace}
\else
\def\EXAGRAM{{\smaller AUTOGRAM}\xspace}
\fi
\def\AUTOGRAM{{\smaller AUTOGRAM}\xspace}

\def\CSV{{\smaller CSV}\xspace}
\def\DAIKON{{\smaller DAIKON}\xspace}

\def\GLADE{{\smaller GLADE}\xspace}
\def\INI{{\smaller INI}\xspace}
\def\JAVA{Java\xspace}

\def\JSON{{\smaller JSON}\xspace}
\def\PDF{{\smaller PDF}\xspace}

\def\URL{{\smaller URL}\xspace}
\def\URLs{{\smaller URLs}\xspace}

\definecolor{rltred}{rgb}{0.5,0,0}
\definecolor{rltgreen}{rgb}{0,0.5,0}
\definecolor{rltblue}{rgb}{0,0,0.5}
\definecolor{todoorange}{rgb}{1, 0.8, 0.4}
\definecolor{todoblue}{rgb}{0.4, 0.8, 1.0}
\definecolor{todogreen}{rgb}{0.8, 1.0, 0.4}
\definecolor{todopurple}{rgb}{0.8, 0.2, 0.8}
\definecolor{todocyan}{rgb}{0.6, 1.0, 1.0}
\ifdraft
    \newcommand\remark[1]{%
      \mymarginpar{\raggedright\hbadness=10000\tiny\it #1\par}}
  \newcommandx{\matthias}[2][1=]{\todo[bordercolor=todoblue!80!black,color=todoblue, fancyline, size=\small, inline]{%
  \sf\small\textbf{%
  \ifx&#1&%
  \else%
  #1\raisebox{.1ex}{$\blacktriangleright$}%
  \fi%
  MH:} #2}}
  \newcommandx{\andreas}[2][1=]{\todo[bordercolor=todoorange!80!black,color=todoorange, fancyline, size=\small, inline]{%
  \sf\small\textbf{%
  \ifx&#1&%
  \else%
  #1\raisebox{.1ex}{$\blacktriangleright$}%
  \fi%
  AZ:} #2}}

  \newcommandx{\alex}[2][1=]{\todo[bordercolor=todoorange!80!black,color=todocyan, fancyline, size=\small, inline]{%
  \sf\small\textbf{%
  \ifx&#1&%
  \else%
  #1\raisebox{.1ex}{$\blacktriangleright$}%
  \fi%
  AK:} #2}}

  \newcommandx{\all}[2][1=]{\todo[bordercolor=todocyan!80!black,color=todocyan, fancyline, size=\small, inline]{%
  \sf\small\textbf{%
  \ifx&#1&%
  \else%
  \textbf{[#1]}%
  \fi%
  } #2}}
  
  \newcommandx{\reviewer}[1]{\todo[bordercolor=rltred!80!black,color=red, fancyline, size=\small, inline]{\sf\small\textbf{RV:} #1}}
\else
    \newcommand\remark[1]{}
    \newcommandx{\matthias}[1]{}
    \newcommandx{\andreas}[1]{}
    \newcommandx{\alex}[1]{}
    \newcommandx{\reviewer}[1]{}
    \newcommandx{\all}[1]{}
\fi

\newtheorem{definition}{Definition}
\newtheorem{lemma}{Lemma}

\begin{document}
%
\title{Active Learning of Input Grammars}

\ifdoubleblind
\else
\author{\IEEEauthorblockN{Matthias Höschele}
\IEEEauthorblockA{Saarland University\\
Saarbr\"ucken Informatics Campus \\ Saarbr\"ucken, Germany\\
Email: hoeschele@cs.uni-saarland.de}
\and
\IEEEauthorblockN{Alexander Kampmann}
\IEEEauthorblockA{Saarland University\\
Saarbr\"ucken Informatics Campus \\ Saarbr\"ucken, Germany\\
Email: kampmann@st.cs.uni-saarland.de}
\and
\IEEEauthorblockN{Andreas Zeller}
\IEEEauthorblockA{Saarland University\\
Saarbr\"ucken Informatics Campus \\ Saarbr\"ucken, Germany\\
Email: zeller@cs.uni-saarland.de}}
\fi


%


\maketitle


\begin{abstract}
Knowing the precise format of a program's input is 
a necessary prerequisite for systematic testing.
Given a program and a small set of sample inputs, we (1)~\emph{track the data flow of inputs} to aggregate input fragments that share the same data flow through program execution into lexical and syntactic entities; (2)~assign these entities \emph{names} that are based on the associated variable and function identifiers; and (3)~systematically \emph{generalize production rules} by means of membership queries.  As a result, we need only a \emph{minimal set of sample inputs} to obtain human-readable \emph{context-free grammars} that reflect valid input structure.  In our evaluation on inputs like \URLs, spreadsheets, or configuration files, our \EXAGRAM prototype obtains input grammars that are both accurate and very readable---and that can be directly fed into test generators for comprehensive automated testing.
\end{abstract}


%
\IEEEpeerreviewmaketitle

\section{Introduction}
\label{introduction}

Systematic testing of any program requires knowledge what makes a valid input for the program---formally, the \emph{language} accepted by the program.  To this end, computer science has introduced \emph{formal languages} including regular expressions, and context-free grammars, which are long part of the standard computer science education canon.  The problem of automatically \emph{inferring} a language from a set of samples is well-known in computer linguistics as well as for compression algorithms.  Inferring the input language for a given \emph{program,} however, only recently has attracted the attention of researchers.  The \AUTOGRAM tool~\cite{hoeschele2016} observes the dynamic data flow of an input through the program to derive a matching input grammar.  The \GLADE tool~\cite{bastani2016} uses membership queries to refine a grammar from a set of inputs.  The \emph{Learn\&Fuzz} approach by Godefroid et al.~\cite{godefroid2017} uses machine learning to infer structural properties, also from a set of inputs.  All these approaches are motivated by using the inferred grammars for automatic \emph{test generation:} Given an inferred grammar, one can easily derive a \emph{producer} that ``fuzzes'' the program with millions of grammar-conforming inputs.

One weakness of all these approaches is that they rely on a \emph{set of available input samples;} and the variety and quality of these input samples determines the features of the resulting language model.  If all input samples only contain positive integers, for instance, the resulting language model will never encompass negative integers, or floating-point numbers. This is a general problem of all approaches for language induction and invariant inference.

In this paper, we present an approach that combines an existing grammar learning technique, \AUTOGRAM, with \emph{active membership queries} to systematically \emph{generalize learned grammars.}  Our \EXAGRAM prototype starts with a program and a minimum of input samples, barely covering the essential features of the input language.  Then, \EXAGRAM will systematically produce possible generalizations to check whether they would be accepted as well; and if so, extend the input grammar accordingly.  The result is a grammar that is both general and accurate, and which can be directly fed into test generators for comprehensive automated testing.

Let us illustrate these abstract concepts with an example.  Given the input in \autoref{fig:url-input}, \EXAGRAM will first derive the initial concrete grammar in \autoref{fig:url-concrete-grammar}.  To do so, \EXAGRAM follows the approach pioneered by \AUTOGRAM (\autoref{fig:overview}), namely following the flow of input fragments into individual functions and variables of the program, and adopting their identifiers to derive the names of nonterminals. The \<'http'> characters, for instance, are stored in a variable named \<protocol>; this results in the \lterm{PROTOCOL} grammar rule.

\begin{figure}[t]
\begin{sample}
http://user:pass@www.google.com:80/command
?foo=bar&lorem=ipsum#fragment
\end{sample}
\caption{Sample \URL input}
\label{fig:url-input}
\end{figure}

\begin{figure}[t]\centering
\begin{grammar}
\lterm{SPEC} \expandsinto \rterm{STRING} '?' \rterm{QUERY} '#' \rterm{REF}
\lterm{STRING} \expandsinto \rterm{PROTOCOL} '://' \rterm{AUTHORITY} \rterm{PATH}
\lterm{AUTHORITY} \expandsinto \rterm{USERINFO} '@' \rterm{HOST} ':' \rterm{PORT}
\lterm{PROTOCOL} \expandsinto 'http'
\lterm{USERINFO} \expandsinto 'user:pass'
\lterm{HOST} \expandsinto 'www.google.com'
\lterm{PORT} \expandsinto '80'
\lterm{PATH} \expandsinto '/command'
\lterm{QUERY} \expandsinto 'foo=bar&lorem=ipsum'
\lterm{REF} \expandsinto 'fragment'
\end{grammar}
\caption{Initial concrete grammar derived by \EXAGRAM from \<java.net.URL> processing the input from \autoref{fig:url-input}.  Using membership queries, \EXAGRAM generalizes this into the grammar in~\autoref{fig:url-general-grammar}.}
\label{fig:url-concrete-grammar}
\vspace{-\baselineskip}
\end{figure}

\begin{figure*}[t]
\includegraphics*[width=\textwidth]{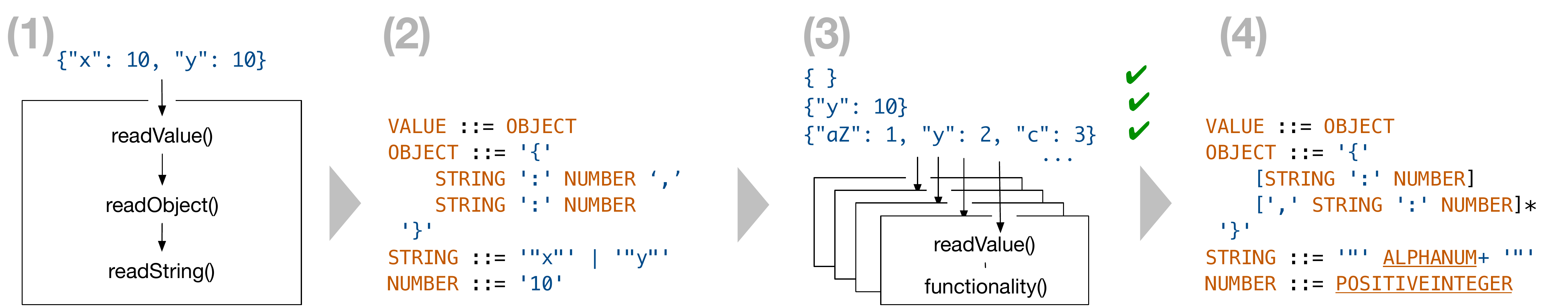}
\vspace{-1.5\baselineskip}
\caption{How \EXAGRAM works.  Given a program and a sample input~(1), \EXAGRAM tracks the flow of input characters through the program to derive a concrete grammar~(2).  In each grammar rule, \EXAGRAM then determines whether parts are optional, can be repeated, or can be generalized to predefined items by querying the program whether inputs produced from the generalization step are still valid~(3).  The resulting grammar~(4) generalizes the original grammar and can quickly produce millions of tests for the program under test---and all other programs with the same input language.}
\label{fig:overview}
\end{figure*}

For each rule in the grammar, \EXAGRAM now attempts to \emph{generalize} it.  To this end, \EXAGRAM applies three rules:

\begin{description}
\item[Detecting optional fragments.]  For each fragment of a rule, \EXAGRAM determines whether it is optional---that is, whether removing it still results in a valid input.  This is decided by a \emph{membership query} with the program whose grammar is to be learned.

In our example (\autoref{fig:url-concrete-grammar}), \EXAGRAM can determine that the fragments \rterm{PATH}, \rterm{USERINFO} '@', 
\<':' \rterm{PORT}>, \<'?' \rterm{QUERY}>, and \<'\#' \rterm{REF}> are all optional, because \<http://www.google.com> is still a valid URL accepted by the \<java.net.URL> parser.
The \rterm{PROTOCOL} and \rterm{HOST} parts are not optional, though, because neither \<www.google.com> (no protocol) nor \<http://> (no host) are valid \URLs.

\item[Detecting repetitions.]  If \EXAGRAM detects that some item is repeated (i.e.\ it occurs at least twice in a rule), it also attempts mailto generalize the repetition by synthesizing alternate numbers of occurrences, and again extending the grammar accordingly if the input is valid.  If an item is found to be repeatable zero to five times, \EXAGRAM assumes an infinite number of repetitions.  In \autoref{fig:overview}, we show how \EXAGRAM generalizes possible repetitions given a \JSON sample.

\item[Generalizing grammar items.]  \EXAGRAM maintains a user-configurable list of \emph{commonly occurring grammar items,} encompassing identifiers, numbers, and strings.  For each string in a rule that matches a predefined grammar item, \EXAGRAM determines the most general item that still produces valid inputs, but where the next most general item is invalid.

Using this rule, \EXAGRAM can determine that \lterm{PORT} generalizes to any natural number (as a sequence of digits), as all of these would be accepted; however, generalizations such as negative numbers or floating point numbers would be rejected.  \lterm{HOST} generalizes to letters, numbers, and dots, but no other characters.  \lterm{PATH} generalizes to an arbitrary sequence of printable characters, as do \lterm{QUERY} and \lterm{REF}.  The \lterm{PROTOCOL} fragment, however, cannot be generalized to anything beyond \<'http'>.
\end{description}

As a result, \EXAGRAM produces the generalized grammar shown in \autoref{fig:url-general-grammar}.  This grammar identifies all optional parts (shown in square brackets), and generalizes all strings to the most general matching items (shown as underlined).  The resulting grammar effectively describes valid \<'http'> \URLs; adding one input sample each for additional protocols (\<'ftp'>, \<'mailto'>, \<'file'>, \dots) easily extends the set towards these protocols, too.

Being able to infer an accurate input grammar from a minimum of samples can be a great help in understanding programs and their input formats.  First and foremost, though, our approach opens the door for \emph{widespread, varied, and fully automatic robustness testing (``fuzzing'') of arbitrary programs that process serial input.}  In our case, given the inferred \URL grammar as a producer, we can easily generate millions of valid and varied \URL inputs, which can then be fed into any (now uninstrumented) \URL parser; other input and file formats would automatically be derived and tested in a similar fashion.
In one sentence, \emph{given only a bare minimum of samples, \EXAGRAM produces a grammar that easily allows the creation of millions and millions of valid inputs.}

\begin{figure}[t]\centering
\begin{grammar}
\lterm{SPEC} \expandsinto \rterm{STRING} [['?' [\rterm{QUERY}]] ['#' [\rterm{REF}]]]
\lterm{STRING} \expandsinto \rterm{PROTOCOL} '://' \rterm{AUTHORITY} [\rterm{PATH}]
\lterm{AUTHORITY} \expandsinto [\rterm{USERINFO} '@'] \rterm{HOST}
     [':' \rterm{PORT}]
\lterm{PROTOCOL} \expandsinto 'http'
\lterm{USERINFO} \expandsinto 'user:pass'
\lterm{HOST} \expandsinto \pterm{HOSTNAME}
\lterm{PORT} \expandsinto \pterm{DIGITS}
\lterm{PATH} \expandsinto \pterm{ABSOLUTEPATH}
\lterm{QUERY} \expandsinto \pterm{ALPHANUMWITHSPECIALS}
\lterm{REF} \expandsinto \pterm{ALPHANUMWITHSPECIALS}
\end{grammar}
\caption{\EXAGRAM grammar generalizing over \autoref{fig:url-concrete-grammar}.  Optional parts are enclosed in [\dots]. Predefined nonterminals (\autoref{fig:predefined}) are underlined.}
\label{fig:url-general-grammar}
\vspace{-\baselineskip}
\end{figure}

In the remainder of this paper, we detail the following contributions:
\begin{itemize}
\item After discussing the state of the art in inferring input grammars (\autoref{sec:background}), \autoref{sec:grammar-inference} contributes a formalization of how \EXAGRAM determines grammars from inputs, extending the informal description in the new idea paper of H\"oschele and Zeller~\cite{hoeschele2016}.
\item In \autoref{sec:generalizing}, we describe the generalization steps specific to \EXAGRAM, showing how \EXAGRAM can derive general grammars from a minimum of sample inputs.  No other technique can infer general input models from single samples alone.
\item \autoref{sec:evaluation} evaluates \EXAGRAM in terms of completeness and soundness of the inferred grammars; we show that the grammars produced are both complete and sound.
\end{itemize}
The paper closes with conclusion and future work in \autoref{sec:conclusion}. The \EXAGRAM prototype and all experimental data is available for research purposes upon request.

\begin{figure*}
\begin{tabular}{@{}>{\scriptsize}r@{\quad}>{\footnotesize\ttfamily\bfseries\color{DarkBlue}}ll>{\footnotesize}l@{}}
0 & http://user:pass@www.google.com:80/command?foo=bar\dots & \PASS & initial input from \autoref{fig:url-input} \\
1 & http://user:pass@www.google.com:\underline{10}/command?foo=bar\dots & \PASS & generalize \lterm{PORT} to \pterm{POSITIVEINTEGER} \\
2 & http://user:pass@www.google.com:\underline{01}/command?foo=bar\dots & \PASS & generalize \lterm{PORT} to \pterm{DIGITS} \\
4 & http://user:pass@www.google.com:\underline{Az1}/command?foo=bar\dots & \FAIL & can't generalize \lterm{PORT} to \pterm{ALPHANUMS} \\
5 & 
http://user:pass@\underline{sub.domain-0.top}:80/command?foo=bar\dots & \PASS & generalize \lterm{HOST} to \pterm{HOSTNAME} \\
6 & http://user:pass@\underline{Az1;:-+=!?*()/\#\$\%\&@}:80/command?foo=bar\dots & \FAIL & can't generalize to \pterm{ALPHANUMWITHSPECIALS} \\
7 &
http://\dots:80/command?\dots{}bar\#\underline{Az1;:-+=!?*()/\#\$\%\&} & \PASS & generalize \lterm{REF} to \pterm{ALPHANUMWITHSPECIALS} \\
8 &
http://\dots:80/command?\dots{}bar\#\underline{Az1;:-+=!?*()/\#\$\%\& \bs{}n\bs{}t\bs{}r} & \FAIL & can't generalize \lterm{REF} to \pterm{PRINTABLES} \\
9 &
http://\underline{Az1;:-+=!?*()/\#\$\%\&@}@www.google.com:80/command?\dots & \FAIL & can't generalize \lterm{USERINFO} to \pterm{PRINTABLES} \\
10 &
\underline{Az}://user:pass@www.google.com:80/command?foo=bar\dots & \FAIL & can't generalize \lterm{PROTOCOL} to \pterm{ALPHAS} \\
11 & 
http://\dots:80/command?\underline{Az1;:-+=!?*()/\#\$\%\&@}\#fragment & \PASS & generalize \lterm{QUERY} to \pterm{ALPHANUMWITHSPECIALS} \\
12 &
http://\dots:80/command?\underline{Az1;:-+=!?*()/\#\$\%\&@ \bs{}n\bs{}t\bs{}r}\#fragment & \FAIL & can't generalize \lterm{QUERY} to \pterm{PRINTABLES} \\
13 & http://\dots:80\underline{/some/path-to/file0123.ext}?foo=bar\dots & \PASS & generalize \lterm{PATH} to \pterm{ABSOLUTEPATH} \\
14 &
http://\dots:80\underline{some/path-to/file0123.ext}?foo=bar\dots & \FAIL & can't generalize \lterm{PATH} to \pterm{PATH}
\end{tabular}
\caption{Refining the grammar through membership queries.  For each rule, \EXAGRAM synthesizes input fragments (underlined) that may further generalize the rule.  By querying the program for whether a synthesized input is valid (\PASS) or not (\FAIL), \EXAGRAM can systematically generalize the concrete grammar from \autoref{fig:url-concrete-grammar} to the abstract grammar in \autoref{fig:url-general-grammar}.}
\label{fig:queries}
\end{figure*}

\section{Background}
\label{sec:background}

\EXAGRAM contributes to three fields: language induction, test generation, and specification mining.

\begin{description}
\item[Language Induction.] 
\EXAGRAM addresses the problem of \emph{language induction}, that is, finding an abstraction that best describes a set of inputs.  Traditionally, language induction was motivated from natural language processing, learning from (typically tagged) sets of concrete inputs; the recent book by Heinz et~al.~\cite{heinz2016} very well represents the state of the art in this field.  

Only recently have researchers turned towards learning the languages of \emph{program inputs.}  The first approach to infer context-free grammars from programs is \AUTOGRAM by H\"oschele and Zeller~\cite{hoeschele2016}; given a program and a set of sample inputs, \AUTOGRAM dynamically tracks the data flow of inputs through the program and aggregates inputs that share the same path into syntactical entities, resulting in well-structured, very readable grammars.  \EXAGRAM follows the \AUTOGRAM approach to infer grammars.

The \GLADE tool by Bastani et~al.~\cite{bastani2016} starts from a set of samples and than uses \emph{membership queries}---that is, querying the program under test whether synthesized inputs are part of the language or not---to derive a context-free grammar that is useful for later fuzzing.  The \emph{Learn\&Fuzz} approach by Godefroid et al.~\cite{godefroid2017} uses machine learning to infer context-free structures from program inputs for better fuzzing, and is shown to be applicable to formats as complex as \PDF.  Compared to \AUTOGRAM (and \EXAGRAM), neither \GLADE nor \emph{Learn\&Fuzz} need or make use of program structure.  This results in a simpler application, but also possibly less structured and less readable grammars; for the purpose of fuzzing, however, these deficits need not matter.

Whether they focus on natural language or program input, though, all of these approaches rely on the presence of \emph{a large set of sample inputs} to induce the language from; consequently, features and quality of the resulting grammars  depend on the variability of the input samples.  \EXAGRAM is unique in requiring only a minimal set of samples; instead, it leverages \emph{active learning} to systematically generalize the well-structured grammar induced by a single sample input already.

\item[Test Generation.]
Techniques for \emph{test generation} also have seen a rise in popularity in the last decade.  For small programs, and at the \emph{unit level,} a wide range of techniques focuses on establishing a wide \emph{variance} between generated runs, reaching branch coverage through \emph{symbolic constraint solving}~\cite{cadar2008,tillmann2008} or \emph{search-based approaches}~\cite{fraser2011}.  For the \emph{system level} of larger programs, the length of paths typically prohibits pure constraint-solving and search-based approaches.  To scale, test generators thus either need a \emph{model} of how the input is structured~\cite{godefroid2008grammar,holler2012fuzzing}, or again \emph{sample inputs} to mutate.  Given these, tools can again use search-based~\cite{afl} or symbolic approaches~\cite{godefroid2008sage} to achieve coverage.

\item[Specification Mining.]
The grammars inferred by \EXAGRAM can also be interpreted as specifications, notably as system-level \emph{preconditions.} By learning from executions, \EXAGRAM is similar in spirit to the \DAIKON approach by Ernst et al.~\cite{ernst2001}, which infers pre- and postconditions at the function level.  In contrast to \DAIKON, though, \EXAGRAM needs only a minimum of sample inputs, as it generalizes these automatically.  In the absence of function-level preconditions, such \emph{active learning} is only possible at the system level, as the program under test decides whether an input is valid or not.  Generally speaking, tools like \EXAGRAM can dramatically improve dynamic specification mining (including grammar mining), as they provide a large variety of inputs (and thus executions) for a large class of programs.
\end{description}

\section{Grammar Inference}
\label{sec:grammar-inference}

We start with a formal description of how \EXAGRAM infers grammars from sample inputs.  In short, \EXAGRAM tracks the path of each input character as well as values derived thereof throughout a program execution, annotating variables and functions with the \emph{character intervals} that flow into them (\autoref{sec:tainting}).  These intervals are then arranged to form hierarchical \emph{interval trees} (\autoref{sec:interval-trees}) which reflect the subsumption hierarchy in the grammar.  After resolving possible overlaps (\autoref{sec:overlap}), \EXAGRAM clusters those elements that are processed in the same way by a program (\autoref{sec:clustering}) to finally derive grammars (\autoref{sec:deriving-grammars}).

\subsection{Tainting}
\label{sec:tainting}

The first step in \EXAGRAM is to use \emph{dynamic tainting} to track the dataflow of input characters.

\subsubsection{Machine Model}
\andreas{Can we shorten the machine model to the first definition and proceed directly towards taints?}

\begin{definition}
        A \textit{program} is a mapping $m$ of fully-qualified method names to sequences $\vec{p}$ of program statements.
\end{definition}

\begin{definition}
        A \textit{program state} is a tuple $(f, s)$, where $f$ is a list of tuples $f_i$ of fully-qualified method names $f_i.n$ and a program counter $f_i.pc$. In further discussion, we will write $f_i = (f_i.n, f_i.pc)$, also we will use $f[0] = f^h$ to refer to the top-most tuple in $f$, $f[1]$ for the second tuple from the top and so on. $s$ is a function which maps program variables $v \in V$ to values.
\end{definition}

The JVM is a stack machine, which means e.g. the instruction \verb|iadd| is defined to pop two integers of the stack\footnote{This is the value stack, not the function stack $f$}, add them and push the result to the stack. For our presentation, we will assume that the JVM is a register machine, so our version of \verb|iadd|, written as \reg{p} = \verb|iadd| \reg{a} \reg{b} adds the values in registers \reg{a} and \reg{b} and stores the result in \reg{p}. This does not hurt correctness of our model, because the JVM byte code can be translated to code for a register machine (e.g. \cite{davis2003case}), but it does make the formalization much more readable. Also the formalization can be applied to other instruction sets (e.g. LLVM-IR) without significant changes.

Thereby the set of program variables $V$ contains heap locations, class variables and object variables, as well as local variables and registers. Registers correspond to elements on the stack in executions on the JVM.

If a program $m$ is executed in a program state $(f_0, s_0)$, the input state, this leads to a state $(f_n, s_n)$, the output state. We will denote \textit{program execution} as $(f_0, s_0) \xrightarrow{m} (f_n, s_n)$. A statement consists of an \textit{Opcode}, as defined in the Java Byte Code Specification and parameters, which are part of the program itself.

$\xrightarrow{m}$ is defined in terms of a helper function $\xrightarrow{p}$, which executes a single program statement. Again, $(f_0,s_0)\xrightarrow{p} (f_1,s_1)$ means that execution of the program statement $p$ in program state $(f_0,s_0)$ leads to a program state $(f_1, s_1)$.  With this helper, $\xrightarrow{m}$ is defined as

\begin{align*}
        (f_0,s_0) \xrightarrow {m} (f_n, s_n) \equiv \\
        (f_0, s_0) \xrightarrow{m(f_0^h.n)[f_0^h.pc]} (f_1, s_1) \\
        \xrightarrow{m(f_1^h.n)[f_1^h.pc]} \ldots \xrightarrow{m(f_{n-1}^h.n)[f_{n-1}^h.pc]} (f_n, s_n)
\end{align*}

This definition of $\xrightarrow{m}$ is not total, as the sequence of applications of $\xrightarrow{p}$ may never end, i.e. the program may not terminate.

In further discussion, we will use $\vec{p}$ to refer to the sequence of program statements that is used in a program execution. We write $p_1 < p_2$ for statements $p_1$ and $p_2$ iff $p_1$ occurs before $p_2$ in an execution. Consequently, statements from different executions can not be compared.

Also, we will use $V^p$ for the set of variables $v$ such that $s_0(v) \neq s_1(v)$ in $(f_0, s_0) \xrightarrow{p} (f_1, s_1)$.

We are not going to report the definition of $\xrightarrow{p}$ for all opcodes in the Java Byte Code, as most of them are straight forward. \autoref{rules} provides some examples.

For \textit{branching instructions}, we change the value of $f.pc$, such that the next application of $\xrightarrow{p}$ executes the instruction after the true- or false-side of the branch respectively. \textit{Method invocations} push a new $f^h$ to the list $f$, such that the lookup $m(f_i^h)$ in the next step returns the newly invoked method. \textit{Method return statements} pop the topmost value from the list $f$.

\subsubsection{Taint Propagation}
\label{sec:propagation}

During execution, the program reads from files. We refer to an input byte sequence, the content of an input file, as $\vec{s}$. If there is more than one input file, each file gets an id $j$, we refer to the input from this file as $_j \vec{s} = _j \vec{s}_0, \ldots, _j \vec{s}_m$. In our implementation, we use the file names of the input files rather than the id. 

\begin{definition}
	For each $_j\vec{s}_k$, the tuple $(j, k)$ is a \textit{taint tag} $t \in T$. 	A function $\tau$ which maps variables $v$ to taint tags is a \textit{taint tag mapping}. 
	
	Taint tag members compare by the byte index. That is, $(j, k) \in T < (g, h) \in T$ iff $j == g$ (they are from the same input source) and $k < h$ (the smaller tag is to the left of the larger one). 
	
	In the following taint tags usually appear as sets $t \subseteq T$. The distinction between individual tags and sets of tags is merely a technical one, so we will refer to those sets as taint tags as well. 
	
	A taint tag $t \subseteq T$ is \textit{consecutive with respect to $j$} iff for all $(j, i) < (j, h) \in t$, all $(j, k) \in T$ such that $i < k < h$ it is $(j, k) \in t$. 
\end{definition}

For our grammar mining, there is just one relevant input file. We will then assume that all taint tags refer to this input file and skip mentioning the id of this file. 

\begin{table*}
	\begin{tabular}{p{.5\textwidth}p{.5\textwidth}}
	\opcode{\reg{p} = \text{iload variable}}
	{ s_1(v) = \begin{cases} s_0(\text{variable}) & \text{if } v == \reg{p} \\ s_0(v) & \text{else} \end{cases} }
	{ \tau_1(v) = \begin{cases} \tau_0(\text{variable}) & \text{if } v == \reg{p} \\ \tau_0(v) & \text{else} \end{cases} }
	{ f_1[y] = \begin{cases}
					(f_0[y].n, f_0[y].pc + 1) & \text{if } y == 0 \\
					f_0[y] & \text{else} \end{cases} }
	\opcode{ \reg{p} = \text{ iadd } \reg{a} \reg{b} }
	{ s_1(v) = \begin{cases} s_0(\reg{a}) + s_0(\reg{b}) & \text{if } v == \reg{p} \\ s_0(v) & \text{else} \end{cases} }
	{ \tau_1(v) = \begin{cases} \tau_0(\reg{p}) & \text{if } v == \reg{p} \\ \tau_0(v) & \text{else} \end{cases} }
	{ f_1[y] = \begin{cases}
			(f_0[y].n, f_0[y].pc + 1) & \text{if } y == 0 \\
			f_0[y] & \text{else} \\
	\end{cases} }
	\opcode{ \text{invokevirtual \$method} \allowbreak \reg{callee} \reg{arg}_0 \ldots \reg{arg}_n }
	{ s_1(v) = \begin{cases} s_0(\reg{callee}) & \parbox[t]{.2\textwidth}{if $v$ is \texttt{this} in the newly called function.} \\ 
			s_0(\reg{arg}_i) & \parbox[t]{.2\textwidth}{ if $v$ is the $i$-th argument in the newly called function.} \\ 
			s_0(v) & \text{else} \\ 
	\end{cases} }
	{ \tau_1(v) = \begin{cases} \tau_0(\reg{callee}) & \parbox[t]{.2\textwidth}{if $v$ is \texttt{this} in the newly called function.} \\ 
			\tau_0(\reg{arg}_i) & \parbox[t]{.2\textwidth}{if $v$ is the $i$-th argument in the newly called function.} \\ 
			\tau_0(v) & \text{else}
	\end{cases} }
	{ f_1[y] = \begin{cases} (\text{\$method}, 0) & \text{if } y == 0 \\ 
			(f_0[0].n, f_0[0].pc + 1) & \text{if } y == 1 \\ 
			f_0[y-1] & \text{else}  \end{cases} }
	\opcode { \text{ireturn} \reg{p} }
	{ s_1(v) = \begin{cases} s_0(\reg{p}) & \parbox[t]{.2\textwidth}{if $v$ is the topmost stack frame.}  \\ 
			s_0(v) & \text{else} 
	\end{cases} }
	{ \tau_1(v) = \begin{cases} \tau_0(\reg{p}) & \parbox[t]{.2\textwidth}{if $v$ is the topmost stack frame.}  \\ 
			\tau_0(v) & \text{else}
	\end{cases} }
	{ f_1[y] = f_0[y+1] }

	\opcode { \text{ifeq} \reg{p} \text{ \$target} }
	{s_1 = s_0}
	{\tau_1 = \tau_0}
	{f_1[y] = \begin{cases}
			(f_0[y].n, \text{\$target}) & \text{if } y == 0 \text{ and } s_0[\reg{p}] == 0 \\
			(f_0[y].n, f_0[y].pc + 1) & \text{if } y == 0 \text{ and } s_0[\reg{p}] != 0 \\
			f_0[y] & \text{else} \end{cases}}
\end{tabular}
	
\alex{Basically, all other instructions are more or less equivalent to one of those. Maybe we could add putfield or getfield. Maybe array handling, but this is mostly interesting due to the memory. And memory is kind of hand-wavy anyways. }

\caption{Definitions of the most common JVM instructions.}
\label{rules}
\end{table*}

For tainting, program states are extended with a taint tag mapping, and the program execution function $\xrightarrow{p}$ is extended to also update the taint tag mapping. \autoref{rules} gives the semantics for the most common JVM instructions. In all cases, $\tau$ is defined to reflect the updates to $s$. 

For the sake of readability, we will use $\tau(V) = \bigcup_{v \in V} \tau(v)$ for the taint tags of all variables in a set $V$. For a sequence $\vec{p}$ of statements that were executed in a program run, We will use $\tau(p) = \tau(V^p)$, and $\tau(\vec{p}) = \bigcup_{p \in \vec{p}} \tau(p)$. 

\subsection{Interval Trees}
\label{sec:interval-trees}

\begin{figure}[t]\centering
\includegraphics*[width=\linewidth]{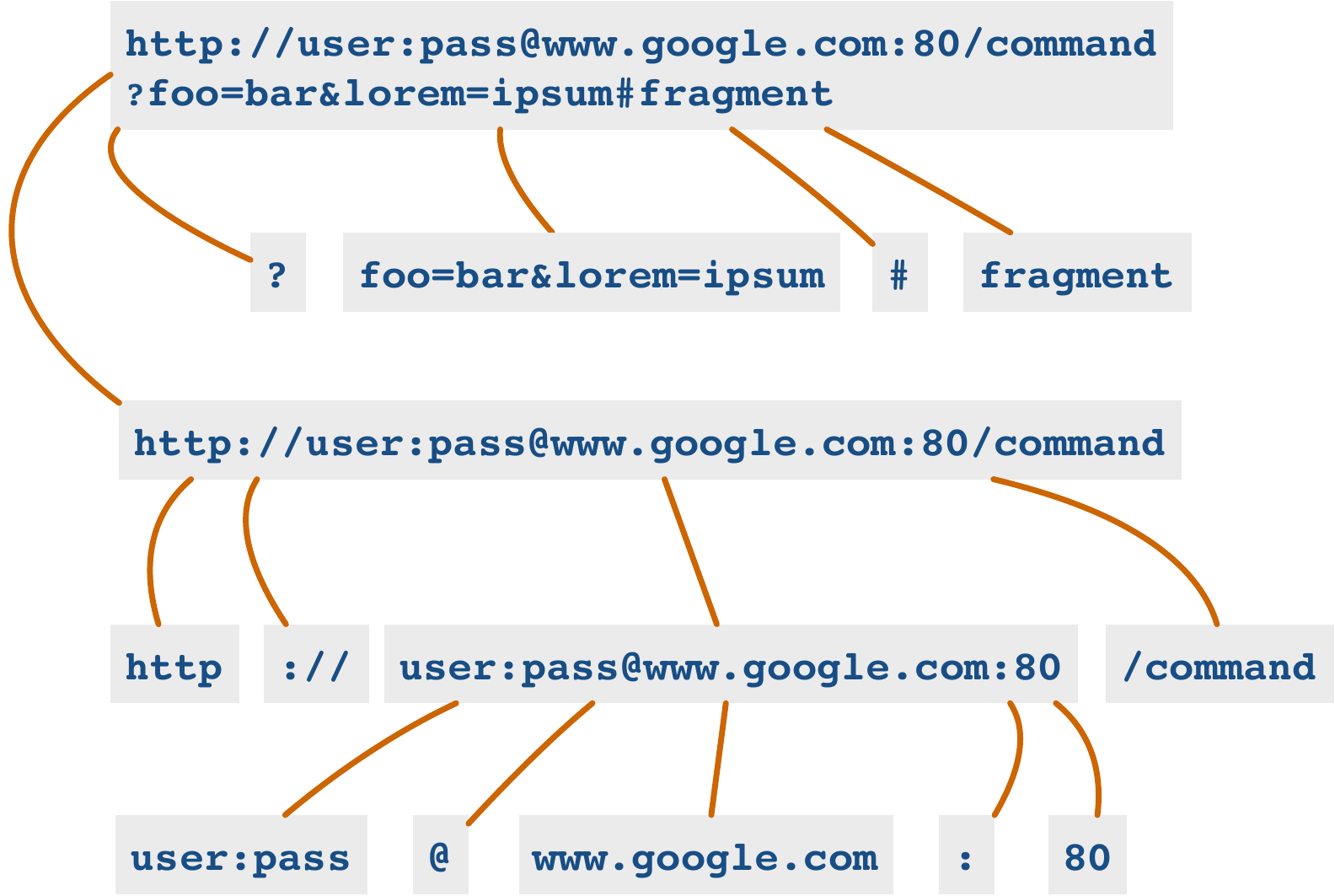}
\caption{Interval tree for the \URL example in \autoref{fig:url-input}}
\label{fig:interval-tree}
\end{figure}

After obtaining the taints for all variables during an execution, the next step in \EXAGRAM is to create a structural representation of the provided sample inputs.  To this end, we create \emph{tree representations} that are approximations of the parse trees which we will later use to infer a grammar.  An example of such an \emph{interval tree} is shown in \autoref{fig:interval-tree}, representing the decomposition of the \URL in \autoref{fig:url-input} by \<java.net.URL>.  We now show how to obtain such an interval tree from the program execution.

Let $\vec{p}$ be the sequence of statements that was executed in a program run. Let $\vec{p}_m = \{p_u, \ldots, p_w\} \subseteq \vec{p}$ be the subsequence of $\vec{p}$ such that $p_u$ is the first statement executed in a method~$m$ and $p_w$ is the return of the same invocation of~$m$. 

Due to the semantics of Java, for any two method invocations $m_1$ and $m_2$, it is either $\vec{p}_{m_1} \subset \vec{p}_{m_2}$, $\vec{p}_{m_2} \subset \vec{p}_{m_1}$ or $\vec{p}_{m_1} \cap \vec{p}_{m_2} = \emptyset$. That is because if a method~$m$ calls another method~$m'$, then $m'$~has to return before~$m$ can return.

We now associate each method with the \emph{characters it processes.}  First, we extend the definition of taint tag mappings from variables to statements and methods:
\begin{definition}
Let $V$ be the set of variables during the execution of a statement~$\vec{p}$, then $\tau(\vec{p}) = \bigcup_{v \in V}\tau(v)$.
\end{definition}
\begin{definition}
Let $P \subseteq \vec{p}$ be the set of statements executed in method~$m$, then $\tau(m) = \bigcup_{p \in P}\tau(p)$.
\end{definition}

Next, we assign each method a \emph{consecutive interval}~$I_m$ between the first and the last character processed in this method:
\begin{definition}
\label{def:interval}
Let $I'_m = \tau(\vec{p}_m)$ be the part of the input that has been processed within the method invocation~$m$.  Then, the \emph{method input interval} of~$m$ is
$$
		I_m = [\min(I'_m), \max(I'_m)]
$$
\end{definition}

\begin{lemma}
In Definition~\autoref{def:interval} $I'_m$ is consecutive $\Leftrightarrow$ $I_m = I'_m$. 
\end{lemma}
\begin{proof}
If $I'_m$ is consecutive, then $I_m = I'_m$ holds by construction.  If $I_m = I'_m$, then $I_m'$ is consecutive because $I_m$ is consecutive by construction.
\end{proof}

\begin{definition}
Let a block $b$ be a parameter or return of a method, a load or store of a field or array, or a sequence of method invocations $b = m_1, \dots, m_l$ such that $\forall m_k, I = I'_{m_k}$, $m_{k-1}$ is the caller of $m_k$ and there is no method invocation $m'$ with $I = I'_{m'}$ such that $m'$ is the caller of $m_1$ or the callee of $m_l$. We extend the definition of $\tau$ to blocks such that $\tau(b) = \bigcup_{p \in b}\tau(p)$. We only consider blocks $b$ such that $\tau(b)$ is consecutive.
\end{definition}

\begin{definition}
Two blocks $b$ and $b'$ are similar if they are both the same parameter, return, a load or store of a field or array in and off the same method. If both blocks are sequences of method calls $b=b_1, \dots, b_l$ and $b'=b'_1, \dots, b'_l$ they are similar if for all $1 \leq i \leq l$, $b_i$ and $b'_i$ are calls of the same method.
\end{definition}

As method calls in a program run form a tree, those intervals can be arranged in a so-called \emph{interval tree.}  

\begin{definition}
For a program execution $\vec{p}$, the \emph{interval tree} of $\vec{p}$ consists of nodes $N_I for$ all intervals $I_b$ of all blocks~$b \in \vec{p}$.  If method~$m_1$ was called by~$m_2$, the interval $I_{m_1}$ is a child of $I_{m_2}$ in the interval tree or $I_{m_1} = I_{m_2}$. We associate each node $N_I$ with the set of blocks $B(N_I)$ that contains all blocks $b$ such that $I_b = I$. 
\end{definition}

In the interval tree in \autoref{fig:interval-tree}, the topmost node has a block containing the constructor \<URL())>.  This constructor called the method \<String.substring()> with \<\small"http://user:pass@www.google.com:80/command\char34> and also called \<URL.isValidProtocol()> with \<"http\char34>.

\begin{definition}
For an interval tree $T$, let $T(m)$ be all intervals that were derived from calls to a method $m$.
\end{definition}

\begin{lemma}
If an interval $I_{m_1}$ is a child of $I_{m_2}$ (that is, $m_1$ was called by $m_2$), then $I_{m_1} \subseteq I_{m_2}$ holds.
\end{lemma}

\begin{proof}
Let $i \in I_{m_1}$. Thus, there is a program statement $p_i$ such that $i \in \tau(p_i)$ and $p_i \in \vec{p}_{m_1}$. $m_1$ was called by $m_2$, thereby $\vec{p}_{m_1} \subset \vec{p}_{m_2}$, but then $p_i \in \vec{p}_{m_2}$ and thereby $i \in \tau(\vec{p}_{m_2}) = I'_{m_2} \subseteq I_{m_2}$.
\end{proof}

Input characters can be processed at multiple program locations, so the converse does not hold.

For building the interval trees that are the input of our grammar learning heuristics we only consider method invocations and other program events like field accesses such that $I'_m$ is consecutive. Let $U$ be the set of all occurring consecutive intervals. We can build a tree by creating nodes $N_I$ for each $I \in U$ with children $N_{I_1} , \dots , N_{I_n}$ such that $\forall I_k \nexists J \subset U, I_k \subset J \subset I$ and $I_k \subset I$.


\subsection{Resolving Overlap}
\label{sec:overlap}


In recursive descent parsers, we observe quite often that intervals which are siblings in the interval tree overlap. This is caused by \textit{lookahead}. As recursive descent parsers usually read input from left to right, this is the most common type of overlap. This leads to a definition of \emph{overlap-free} interval trees. 

\begin{definition}
An interval tree $T$ is \textit{overlap-free} iff for all $I,J \in T$, either $I \cap J = \emptyset$ or $I \subset J \vee J \subset I$. 
\end{definition}

In an overlap-free interval tree, an interval overlaps only with its children, and children are always contained in their parents entirely. 

There is a simple algorithm to derive a overlap-free interval tree from an interval tree. We utilize that recursive descent parsers read their input from left to right, so any overlap occurs between children of the same parent, and the overlap corresponds to lookahead. Thereby, we can resolve overlap in an interval tree as follows: 

First of all, we derive an order on intervals. An interval $I = [r_I, t_I]$ is smaller than an interval $J = [r_J, t_J]$ if $r_I < r_J$ or $r_I = R_J$ and $t_I < t_J$.

For any node $N \in T$ with children $N_I$ and $N_J$ such that $I = [r_I, t_I]$, $J = [r_J, t_J]$,  $r_J < t_I$ and $I < J$, we derive a replacement interval $I' = [r_I, r_J-1]$. We recursively remove all children $N_C$ with $C = [r_C, t_C]$ if $r_C \not\in I'$, or replace them with nodes $N_{C'}$ with $C' = [r_C, t_J-1]$ if $r_C \in I'$.

In cases where there exists blocks in $B(N_I)$ such that they occur after all blocks in $B(N_J)$ in the pre-order of the call tree, we instead derive a replacement node $N_{J'}$ with $J' = [t_I + 1, t_J]$. We recursively remove all children $N_C$ with $C = [r_C, t_C]$ if $t_C \not\in J'$, or replace them with nodes $N_{C'}$ with $C' = [t_I + 1, t_C]$ if $t_C \in J'$. The intuition behind this is to deal with parser implementations that remember the last character and therefore have show patterns that could be interpreted as a sort of lookback similar to lookahead. This allows us to identify input fragments that are later used in the program and avoid splitting them up during the overlap resolving stage.

As another observation, if $\vec{i}$ is valid, $P$ needs to read the entire input. That is because in an $LL(k)$ language, a word with a valid prefix can always be invalidated by an invalid suffix. So $P$ can only accept after it read all of the input. Thereby, for all valid inputs the root of the interval tree is the entire input.

\subsubsection{Fixing Alignment of Last Leaf}

Due to the way interval trees are constructed the last input fragment is likely to be missaligned especially if it is a single character. For inner fragments the resolving of overlap uses the call tree to remove ambiguities and determine where a node should belong. We also use the dynamic call graph to check if we might propagate the leaf $L = [r_L, t_l]$ corresponding to the last input fragment to be a child to a node closer to the root. For an ancestor $I = [r_I, t_I]$ with $r_I < r_L$  and $t_L = t_I$, we check if there is a block $b_L$ of method calls in $B(L)$ such that a block $b_I$ in $B(I)$ contains the caller of $b_L$. In this case we can recursively remove $L$ from all children of $I$ and propagate it to be a direct child of $I$.

\subsection{Clustering Interval Nodes}
\label{sec:clustering}

After the construction of a set of overlap-free interval trees $O$ we are now at the stage to identify syntactic elements. The intuition for this stage is that syntactic elements of the same type or more precisely derivations of the same non-terminal symbols are processed in the same way by a program. This means the corresponding characters will be processed by the same functions and will be stored in the same fields and variables. We will therefore apply a simple clustering to the set $\mathit{Nodes}(O)$ of all nodes of all trees $T \in O$, that groups together nodes with similar labels.

\begin{definition}
A cluster is defined as a pair $C = (N_r, S)$ of a set of interval nodes $S \subset \mathit{Nodes}(O)$ and a representative node $N_r \in S$. All clusters form a partition $P$ of $O$ such that $\forall S_i,S_j \in P, S_i \cap S_j = \emptyset$. Let $\textit{Cluster}(N) = C$ be the cluster such that $N \in S$.
\end{definition}

We implemented this as a greedy algorithm that starts with an initially empty set $P^*$ of clusters.
Our heuristic sequentially processes all nodes $N \in \mathit{Nodes}(O)$ such that it tries to find a cluster $C' = (N'_r, S') \in P^*$ such that $N$ is \textit{similar} to $N'_r$. If a cluster $C'$ could be found, $N$ is added to $S'$ such that $C' = (N'_r, S' \cup \{N\})$. Otherwise we add a new cluster $C_\mathit{new} = (N, \{N\})$ to $P^*$. The \textit{relative similarity} $\mathit{Sim}_{\mathit{limit}}(N_i, N_j)$ of two nodes $N_i$ and $N_j$ is computed by checking for how many of the first $\mathit{limit}$ blocks according to pre-order in the dynamic call tree in  $N_i$ for which we can find similar blocks in $N_j$. We call this number $\mathit{match}_{\mathit{limit}}(N_i, N_j)$ and compute the same with the roles of $N_i$ and $N_j$ reversed as $\mathit{match}_{\mathit{limit}}(N_j, N_i)$. Let $\mathit{blocks}(N)$ be the number of blocks in a node $N$. 

\begin{definition}
The relative similarity is computed as
\begin{multline}
\mathit{Sim}_{\mathit{limit}}(N_i, N_j) = \\ \biggl(\frac{\mathit{match}_{\mathit{limit}}(N_i, N_j)}{\min(\mathit{blocks}(N_i), \mathit{limit})} + \frac{\mathit{match}_{\mathit{limit}}(N_j, N_i)}{\min(\mathit{blocks}(N_j), \mathit{limit})}\biggr) / 2
\end{multline}
\end{definition}

In our experiments we applied this heuristic with $\mathit{limit} = 15$ which showed to be a reasonable number for our subjects. For nodes $N_i$ and $N_j$ a relative similarity value $\mathit{Sim}_{\mathit{limit}}(N_i, N_j) \leq 0.9$ is considered to be sufficient in our experiments in order to add them to the same cluster and therefore being called \textit{similar}.
In future work it might be necessary to define alternative heuristics for similarity, especially if the technique might be applied to parsing code that frequently uses backtracking for which a fixed limited number of blocks might not be sufficient to determine the similarity of nodes since large amounts of blocks might correspond to invocations that have been discarded.

\subsection{Deriving Grammars}
\label{sec:deriving-grammars}

After clustering the nodes of all interval trees as described in \autoref{sec:clustering} the clusters can be used by our heuristics to derive a context-free grammar with non-terminal symbols and the corresponding production rules.

\subsubsection{Complex and Single Character Clusters}

The first observation after the clustering stage is that input fragments consisting of single characters usually end up in clusters with other single characters. They are also usually not similar to nodes that are supposed to belong to the same grammatical categories. The reason for this is that due to lookahead a parser will treat these single character fragments differently than the complex ones. In addition to parsing them as a numeric character, parsers will frequently access them to determine what parsing function should be called next or for common tasks like trying to skip whitespace. This results in additional blocks in the corresponding nodes that have no similar blocks in nodes corresponding to longer fragments.

Therefore we identify a set of complex clusters $\mathit{Complex}(P)$ and a set of single-character clusters $\mathit{SingleChar}(P)$ that only contain nodes $N = [r_N, t_N]$ with $r_N = t_N$ such that $\mathit{Complex}(P) \cup \mathit{SingleChar}(P) = P$. The next steps will only use clusters in $\mathit{Complex}(P)$ to identify non-terminal symbols and derive productions and the knowledge from clusters in $\mathit{SingleChar}(P)$ is integrated into the derived grammar in a post processing step.


\subsubsection{Identifying Single Non-Terminal Substitution}
\label{sec:grammar:sns}

Input languages frequently represent entities from the input domain of a program in a way that closely corresponds to the way these entities are represented in the code. Such entities are usually modeled as structs or classes by programmers and the relationship between those entities will modeled as type relationships. These types therefore are closely related to non-terminal symbols in a formal grammar for the input language. When we look at \JSON as an example, we can see such a correspondence for the symbol \<\textbf{\lterm{VALUE}}>:

\begin{grammar}
\lterm{VALUE} \expandsinto \rterm{OBJECT} | \rterm{ARRAY} | \rterm{STRING} |
          \rterm{TRUE} | \rterm{FALSE} | \rterm{NULL} | \rterm{NUMBER}
\end{grammar}

The symbol \<\textbf{\lterm{VALUE}}> can be substituted for 7 other non-terminal symbols which correspond to subclasses of an abstract class \<JsonValue>. Since our \JSON library implements recursive descent, these values are read by the method \<readValue()> that depending on the next character decides which specialized parsing method like for example \<readArray()> it needs to call. This also means that the first block $b$ in an interval tree node $n$ for such a input fragment will be a sequence of method calls $b = m_1, \dots, m_k$ such that $m_1$ is an invocation of \<readValue()> and $m_2$ is an invocation of one of the specialized parsing methods. According to the heuristics described in \autoref{sec:clustering} all nodes corresponding to input fragments representing arrays are put in the same cluster but they are not similar to nodes for objects or numbers. The fact that arrays and the other entities are all values is not explicitly visible in our structural decomposition at this point in the inference process.

We address this by searching for common prefixes in the first blocks of nodes. Let $C_1 = (N_{C_1}, S_{C_1})$ and $C_2 = (N_{C_2}, S_{C_2})$ be clusters, $b_{N_{C_2}}  = m_{1,1}, \dots, m_{1,k_1}$ and $b_{N_{C_2}}  = m_{2,1}, \dots, m_{2,k_2}$ are the first blocks in $N_{C_1}$ and $N_{C_2}$. If $b_{N_{C_1}}$ and $b_{N_{C_2}}$ have a common prefix of size $j$ we can split all nodes in $C_1$ and $C_2$. For a node $N$ in $C_1$ or $C_2$ let $b = m_1, \dots, m_k$  be the first block in $N_I$ with $I = [r_I, t_I]$. We split $b$ that consist of sequences of calls to the same methods and their postfixes $b_{\mathit{pre}}  = m_{1}, \dots, m_{j}$, $b_{\mathit{post}} = m_{j+1}, \dots, m_{k}$. Using these blocks we derive new nodes $N^{\mathit{pre}}_I$ and $N^{\mathit{post}}_I$ such that $B(N^{\mathit{pre}}_I) = \{b_{\mathit{pre}}\}$ and $B(N^{\mathit{post}}_I) = B(N_I) \setminus b \cup \{b_{\mathit{post}}\}$ with $N^{\mathit{post}}_I$ being the only child of $N^{\mathit{pre}}_I$. We replace $N_I$ in the interval tree and cluster with $N^{\mathit{pre}}_I$ and transfer all children of $N_I$ to $N^{\mathit{post}}_I$. We also create a new cluster $C_{pre}$ that will contain all nodes $N^{\mathit{pre}}_I$.

Using this transformation we can make the relationship between input fragments explicitly visible in form of intermediate node in the interval trees and a common cluster.

\subsubsection{Create Non-Terminal Symbols and Productions}

At this stage we derive the non-terminal symbols from the complex clusters. For each cluster $C \in \mathit{Complex}(P)$ we define a corresponding non-terminal symbol $\mathit{Sym}(C)$. We can derive a simple set of productions by identify all observed substitution sequences from the children of each $N_I$ with $I = [r_I, t_I] \in C$. Let $N'_1, \dots, N'_k$ be the children of $N_I$ ordered by the position of the corresponding input fragments, then $\mathit{Sym}(\mathit{Cluster}(N'_1)), \dots, \mathit{Sym}(\mathit{Cluster}(N'_k))$ is a possible substitution for $\mathit{Sym}(C)$. If $N$ is a leaf the character sequence corresponding to the interval $[r_I, t_I]$ is a possible substitution for $\mathit{Sym}(C)$. Child nodes $N'_j$ that correspond to single characters are instead represented by a corresponding terminal symbol instead of $\mathit{Sym}(\mathit{Cluster}(N'_j))$. Applying these definitions provides us with a preliminary set of production rules that precisely capture the structure of the observed sample inputs and do not include any generalization. The only generalization up to this point comes from the assumptions in \autoref{sec:clustering}.


\subsubsection{Post Processing}

\begin{description}
\item[Merge Symbols] The clustering in \autoref{sec:clustering} and the processing in \autoref{sec:grammar:sns} might still result in similar syntactic elements being part of different clusters. They are therefore represented by different non-terminal symbols. We address this with a tunable heuristic that merges compatible symbols. For clusters $C_i = (N_i, S_i)$ and $C_j = (N_j, S_j)$ let $b_i$ and $b_j$ be the first blocks in $N_i$ and $N_j$. We merge the clusters $C^{\mathit{merge}} = (N_i, S_i \cup S_j)$ if $b_i$ and $b_j$ match exactly or if they are similar enough according to a relaxed heuristic $\mathit{Sim}^{\mathit{prefix}}_{\mathit{limit}}(N_i, N_j)$ that computes a similarity score that allows for partial prefix matches of blocks. The new cluster also is assigned a new symbol $\mathit{Sym}(C^{\mathit{merge}})$ that is substituted for all occurrences of $\mathit{Sym}(C_i)$ and $\mathit{Sym}(C_j)$.

\item[Process Single Characters] Up to this stage we did not account for the possibility that single-character fragments can also be instances of non-terminal symbols. To address this we apply a heuristic $\mathit{Sim}^{\mathit{prefix}\rightarrow}_{\mathit{limit}}(N_j, N^{\mathit{char}})$ that for each cluster $C_1, \dots, C_k$ and a node $N^{\mathit{char}}$ that computes a similarity score that allows for partial prefix matches of blocks and is unidirectional only tries to find similar of blocks from $N_j$ in blocks of $n^{\mathit{char}}$. If the character occurs in a fragment represented by a node in $C_l$, $\mathit{Sim}^{\mathit{prefix}\rightarrow}_{\mathit{limit}}(N_l, N^{\mathit{char}})$ is greater than for all other $N_j$ and also exceeds a configurable threshold, we replace the terminal symbol in the corresponding production with $\mathit{Sym}(C_l)$.
\end{description}


\subsection{Naming Nonterminals}

In order to make it easier to for users to read the grammars learned by \EXAGRAM we try to propose meaningful names for non-terminal symbols. For each symbol $S=\mathit{Sym}(C)$ with a cluster $C = (n_r, S)$ we collect the names from elements of blocks $b \in B(N)$ with $N \in S$, e.g. names of methods which processed the corresponding fragments or parameters and fields in which a value derived from the fragment was stored. For each cluster $C$ we therefore get a multiset $\mathit{Names}(C)$ of strings that we use to propose a name. We implemented a simple heuristic that first filters the strings by removing common prefixes or suffixes like \texttt{get}, \texttt{set} and \texttt{parse} and then identifies the most often occurring substring that is as long as possible.

%
%
%
%

\section{Generalizing Grammars}
\label{sec:generalizing}

Since we aim to learn grammars from very small sample sets down to one single sample input our initially derived grammar will be a very close fit to these samples. In the following we describe several heuristics for generalizing these grammars.

\subsection{Optional Elements}

The first generalization is the identification of optional elements. For all sequences of symbols $S_1, \dots, S_n$ on the right-hand side of productions we try to identify optional subsequences of length $l = n-1$ down to $l = 1$. For each initial sample that contributed to the derivation of this sequence and each subsequence $S_k, \dots, S_{k+l}$ of length $l$ we derive new inputs by omitting the fragments corresponding to $S_k, \dots, S_{k+l}$. We execute the program with theses inputs and check if they are accepted by the program and the data-flow of the fragments following the omitted part has not changed. In that case we consider the subsequence optional and modify the grammar accordingly. In case of our URL example we start with the concrete grammar in \autoref{fig:url-concrete-grammar} and starting at \lterm{SPEC} derive new inputs by omitting decreasingly long subsequences. \autoref{fig:optionality-samples} shows the accepted and rejected inputs that lead to the generalizations in \autoref{fig:url-general-grammar}.

\begin{figure*}
\begin{tabular}{@{}>{\scriptsize}r@{\quad}>{\footnotesize\ttfamily\bfseries\color{DarkBlue}}ll>{\footnotesize}l@{}}
0 & http://user:pass@www.google.com:80/command?foo=bar\&lorem=ipsum\#fragment & \PASS & initial input from \autoref{fig:url-input} \\
1 & http://user:pass@www.google.com:80/command\textcolor{AlmostWhite}{?foo=bar\&lorem=ipsum\#fragment} & \PASS &  \\
3 & \textcolor{AlmostWhite}{http://user:pass@www.google.com:80/command?foo=bar\&lorem=ipsum\#}fragment & \FAIL &  \\
4 & http://user:pass@www.google.com:80/command?\textcolor{AlmostWhite}{foo=bar\&lorem=ipsum\#fragment} & \PASS &  \\
5 & http://user:pass@www.google.com:80/command\textcolor{AlmostWhite}{?foo=bar\&lorem=ipsum\#}fragment & \FAIL &  \\
6 & \textcolor{AlmostWhite}{http://user:pass@www.google.com:80/command?foo=bar\&lorem=ipsum}\#fragment & \FAIL &  \\
7 & http://user:pass@www.google.com:80/command?foo=bar\&lorem=ipsum\textcolor{AlmostWhite}{\#fragment} & \PASS &  \\
& \multicolumn{1}{c}{\vdots} & & \\
15 & \textcolor{AlmostWhite}{http://user:pass@www.google.com:80/command}?foo=bar\&lorem=ipsum\#fragment & \FAIL &  \\
\end{tabular}
\caption{Identifying optional elements through membership queries.  For each rule, \EXAGRAM synthesizes new inputs that omit the fragments corresponding to a subsequence of elements.  By querying the program for whether a synthesized input is valid (\PASS) or not (\FAIL), \EXAGRAM can systematically generalize the concrete grammar from \autoref{fig:url-concrete-grammar} to the abstract grammar in \autoref{fig:url-general-grammar}.}
\label{fig:optionality-samples}
\end{figure*}

\subsection{Generalizing Strings}

\begin{figure}[t]
\begin{grammar}\small
\pterm{DIGIT} ::= /[0-9]/
\pterm{DIGITS} ::= \pterm{DIGIT}+
\pterm{POSITIVEINTEGER} ::= /[1-9]/ [\pterm{DIGITS}]
\pterm{INTEGER} ::= ['-'] \pterm{POSITIVEINTEGER}
\pterm{ALPHA} ::= /[A-Z]/ | /[a-z]/
\pterm{ALPHAS} ::= \pterm{ALPHA}+
\pterm{ALPHANUM} ::= \pterm{ALPHA} | \pterm{DIGIT}
\pterm{ALPHANUMS} ::= \pterm{ALPHANUM}+
\pterm{WHITESPACE} ::= ' ' | '\bs{}t'
\pterm{WHITESPACES} ::= \pterm{WHITESPACE}+
\pterm{WHITESPACENEWLINE} ::= ' ' | '\bs{}t' | '\bs{}n' | '\bs{}r'
\pterm{WHITESPACENEWLINES} ::= \pterm{WHITESPACENEWLINE}+
\pterm{HOSTNAME} ::= (\pterm{ALPHANUM} | '.' | '-')+
\pterm{PATH} ::= (\pterm{ALPHANUM} | '.' | '-' | '/')+
\pterm{ALPHANUMWHITESPACE} ::= \pterm{ALPHA} | \pterm{DIGIT} |
     \pterm{WHITESPACE}
\pterm{ALPHANUMWHITESPACES} ::= \pterm{ALPHANUMWHITESPACE}+
\pterm{ABSOLUTEPATH} ::= '/' [\pterm{PATH}]
\pterm{ALPHANUMWITHSPECIAL} ::= \pterm{ALPHANUM} | ';' | ':' |
     '-' | '+' | '=' | '!' | '*' | '(' |
     ')' | '/' | '$' | '%' | '&' | '@'
\pterm{ALPHANUMWITHSPECIALS} ::= \pterm{ALPHANUMWITHSPECIAL}+
\pterm{ALPHANUMWITHSPECIALWHITESPACE} ::= 
    \pterm{ALPHANUM} | \pterm{WHITESPACE}
\pterm{ALPHANUMWITHSPECIALWHITESPACES} ::=
    \pterm{ALPHANUMWITHSPECIALWHITESPACE}+
\pterm{PRINTABLENEWLINE} ::= \pterm{ALPHANUMWITHSPECIAL} |
    \pterm{WHITESPACENEWLINE}
\pterm{PRINTABLENEWLINES} ::= \pterm{PRINTABLE}+
\pterm{PRINTABLE} ::= \pterm{ALPHANUMWITHSPECIAL} | 
     \pterm{WHITESPACE} | '?' | '#'
\pterm{PRINTABLES} ::= \pterm{PRINTABLE}+
\end{grammar}
\vspace{-0.5\baselineskip}
\caption{Predefined (user-configurable) nonterminals in \AUTOGRAM.}
\label{fig:predefined}
\end{figure}

Another important step is to generalize terminal strings in grammars. When we learn a grammar from an input like in \autoref{fig:url-input} we only observe one specific value for \lterm{HOST}, \lterm{PORT} and other symbols. We try to generalize the observed values to regular expressions. \autoref{fig:predefined} shows a set of non-terminal symbols that define regular languages. The inclusion relationship between these regular languages can be used to derive a directed graph $G_L$ such that for languages $L_1$ and $L_2$, there is a path from $L_1$ to $L_2$ if $L_1 \subset L_2$. For a set $S$ of strings our heuristic starts by identifying the smallest regular language $L$ such that $S \subset L$. In order to check if we can safely generalize $S$ to $L$, we derive new inputs from all occurrences of each $s \in S$ in the initial sample set which are replaced by a representative member of $L$. Similar to the heuristic for optionality, we run the program with these inputs and check if they are accepted and if the data-flow for the unmodified parts of the input remain the same. If this is successful for all derived inputs, we can generalize $S$ to $L$. Our heuristic traverses $G_l$ step by step until the derived inputs are not accepted or show modified data-flow.

For our URL example, we generalize individual non-terminal symbols as illustrated in \autoref{fig:queries}. For the concrete port number \texttt{80} we find the smallest regular language \lterm{POSITIVEINTEGER} that contains this fragment. To confirm that the program will also accept other elements of \lterm{POSITIVEINTEGER} we replace \texttt{80} with \texttt{10} and run the program to check if this derived input is accepted, which in this case is successful. The next possible step is to try and generalize to the next larger regular language which is \lterm{DIGITS}. We again test by replacing \texttt{80} with \texttt{01} which is again successful. However we are not able to generalize further to \lterm{ALPHANUMS} since replacing \texttt{80} with \texttt{Az1} is rejected by the program.

\subsection{Repetitions}

Many context free languages use repeating features like sequences. We implemented a simple heuristic to detect repeating expressions and try to generalize them using the same mechanism that is used for finding optional elements. For each occurrence of a repetition of an expression in the set of sample inputs we derive new inputs in which the repetition is replaced by $n$ occurrences of the expression. We currently do this for $n \in [1,5]$. We run the program with these inputs and check if they are accepted and if the data-flow for the unmodified parts of the input remain the same. If this is successful for all derived inputs we generalize the reputation to a non-empty sequence of the expression of an arbitrary length.

If the repeated expression is a sequence of a non-terminal symbol $S$, we check if all nodes of occurrences of the first or the last element in the initial sample set correspond to the application of a specific production rule while all others use a different production rule. If this is the case we specialize the first or last element and the other occurrences to these production rules. This helps to deal with common patterns like repeating lines, where all lines except the last one are terminated by a newline symbol.

\section{Evaluation}
\label{sec:evaluation}

Let us now demonstrate \EXAGRAM on a set of example subjects.  As test subjects, we use the programs listed in \autoref{tab:subjects}; this extends the set from \AUTOGRAM~\cite{hoeschele2016} with INI4J.

\begin{table*}[t]
\centering
\caption{Test Subjects and Evaluation Results}
\label{tab:subjects}
\begin{tabular}{lllrr}
\textbf{Subject} & \textbf{Data Format} & \textbf{Format Purpose} & \textbf{Soundness} & \textbf{Completeness} \\ \hline
\texttt{java.lang.URL} & \URL & Uniform Resource Locators; used as Web addresses & 100.0\% & 50.3\% \\
Apache Commons CSV\footnotemark[1] & \CSV & Comma-separated values; used in spreadsheets & 100.0\% & 100.0\% \\
\texttt{java.util.Properties} & \JAVA property files & Configuration files for \JAVA programs using key/value pairs & 100.0\% & 13.4\% \\
INI4J\footnotemark[2] & \INI  & Configuration files consisting of sections with key/value pairs & 100.0\% & 16.1\% \\
Minimal JSON\footnotemark[3] & \JSON & Human-readable text to transmit nested data objects & 100.0\% & 100.0\% \\
\end{tabular}
\end{table*}

\footnotetext[1]{https://commons.apache.org/proper/commons-csv/}
\footnotetext[2]{http://ini4j.sourceforge.net/}
\footnotetext[3]{https://github.com/ralfstx/minimal-json}

\subsection{Soundness and Completeness}

We evaluate each grammar by two measures:
\begin{description}
\item[Soundness.]  For each of the grammars produced by \EXAGRAM, we use them as \emph{producers} to derive 20,000 randomly generated strings.  These then serve as inputs for subjects they were derived from.  The \emph{soundness} measure is the percentage of inputs that are accepted as valid.  A 100\% soundness indicates that all inputs generated from the grammar are valid.

\item[Completeness.]  For each of the languages in \autoref{tab:subjects}, we create a \emph{golden grammar}~$g$ based on the language specification.  We test whether the respective \EXAGRAM-generated grammar accepts the 20,000 random strings generated by~$g$.  A 100\% completeness indicates that the grammar encompasses all inputs of the golden grammar.
\end{description}

When the grammars are used for test generation (fuzzing), a high \emph{soundness} translates into a high \emph{test efficiency,} as only few inputs would be rejected.  A high \emph{completeness} correlates with \emph{test effectiveness,} since the grammar would cover more features, and thus exercise more functionality.

\autoref{tab:subjects} summarizes our results for soundness and completeness across all subjects.  We now discuss these in detail.

\subsection{URLs}

In \autoref{fig:url-general-grammar},  we show the \EXAGRAM grammar obtained from the \<java.net.URL> class and the one sample in \autoref{fig:url-input}.  In our evaluation, the inferred \URL grammar is 100\% sound (all inputs derived from it are complete); however, it is only 50.3\% complete.  This is due to the rule for \lterm{USERINFO}, whose format \emph{user}:\emph{password} cannot be generalized using our predefined items; in our evaluation, this leads to every \URL containing a random user/password combination being rejected.  This, however, is a case that can be very easily fixed by either amending the produced grammar or introducing a predefined item for \emph{user}:\emph{password}.

As all grammars produced by \EXAGRAM, the grammar is very readable; this is due to \EXAGRAM naming nonterminals after associated variables and functions.  We have not formally evaluated readability, but we list the grammars as raw outputs such that readers can assess their understandability themselves.

\subsection{Comma-Separated Values}

\begin{figure}[t]
\begin{sampleandarrow}\footnotesize
Firstname;Lastname;Phone
Roger;Smith;34534534
"Anne";"Perkins Watson";"204059"
Leila;Jackson;9569784
Dough;Clinton;1298483
\end{sampleandarrow}
\vspace{-.5\baselineskip}
\begin{grammar}
\lterm{MAIN} \expandsinto \rterm{CSVPARSER} \rterm{NEXTRECORD}*
\lterm{CSVPARSER} \expandsinto \rterm{HASH} \rterm{NEXTTOKEN}*
\lterm{HASH} \expandsinto \pterm{ALPHANUMWHITESPACES}
\lterm{NEXTTOKEN} \expandsinto ';' [\rterm{KEY}] |
     ';' [\rterm{ENCAPSULATEDTOKEN}]
\lterm{ENCAPSULATEDTOKEN} \expandsinto '"' \rterm{KEY} '"'
\lterm{KEY} \expandsinto \pterm{ALPHANUMWHITESPACES}
\lterm{NEXTRECORD} \expandsinto \rterm{NEXTTOKEN_FIRST} \rterm{NEXTTOKEN}*
\lterm{NEXTTOKEN_FIRST} \expandsinto '\bs{}n' [\rterm{KEY}] |
     '\bs{}n' [\rterm{ENCAPSULATEDTOKEN}]
\end{grammar}
\caption{\CSV grammar generalized and derived from a single sample.}
\label{fig:csv-grammar}
\end{figure}

The Apache Commons CSV reader uses a pure lexical processing for its input, which is also reflected in the resulting grammar (\autoref{fig:csv-grammar}). However, \EXAGRAM nicely identifies that values are optional and that strings can contain arbitrary characters, features not present in the original sample.  The grammar is both 100\% sound and complete.

\subsection{Java Properties}

\begin{figure}[t]
\begin{sampleandarrow}\footnotesize
Version=1
WorkingDir=mydir
# comment 
User=Bob
Password=12345
\end{sampleandarrow}
\vspace{-.5\baselineskip}
\begin{grammar}
\lterm{MAIN} ::= \rterm{LOAD}
\lterm{LOAD} ::= \rterm{LINE}* \rterm{LINE_LAST}
\lterm{LINE} ::= [\rterm{S2}] '\bs{}n'
\lterm{S2} ::= [\rterm{S3}] \rterm{ARG} ' '* '=' [' '* \rterm{ARG}]
\lterm{S3} ::= '# comment \bs{}n'
\lterm{ARG} ::= \pterm{ALPHANUMWHITESPACES}
\lterm{LINE_LAST} ::= \rterm{S2}
\end{grammar}
\caption{Java properties grammar generalized and derived from a single sample.}
\label{fig:properties-grammar}
\end{figure}

The Java properties parser also is a simple scanner, resulting in nonterminals for which \EXAGRAM cannot find a name (\lterm{S2} and \lterm{S3}).  The \emph{key=value} structure is well identified, as are the optional values (\lterm{S2}).  The grammar is 100\% sound; the comment (\;term{S3}), however, cannot be generalized further by \EXAGRAM, resulting in a low completeness if 13.4\% in our evaluation, as 86.6\% of golden inputs would sport different comments.  Again, this is very easy to amend.

\subsection{INI Files}

\begin{figure}[t]
\begin{sampleandarrow}\footnotesize
[Application]
Version = 1
WorkingDir = mydir
[User]
User = Bob
Password = 12345
\end{sampleandarrow}
\vspace{-.5\baselineskip}
\begin{grammar}
\lterm{LOAD} ::= \rterm{LINE_FIRST} \rterm{LINE}* \rterm{LINE1}
\lterm{LINE_FIRST} ::= \rterm{LINE2_FIRST} '\bs{}n'
\lterm{LINE2_FIRST} ::= '[' \rterm{SECTION} ']'
\lterm{SECTION} ::= \pterm{ALPHANUMWHITESPACES}
\lterm{LINE} ::= \rterm{LINE2} '\bs{}n'
\lterm{LINE2} ::= '[' \rterm{SECTION} ']' | 
    \rterm{KEY} ' '* '=' [' '*  [\rterm{VALUE}]]
\lterm{KEY} ::= \pterm{ALPHANUMWHITESPACES}
\lterm{VALUE} ::= \pterm{ALPHANUMWHITESPACES}
\lterm{LINE1} ::= \rterm{KEY} ' '* '=' [' '* [\rterm{VALUE}]]
\end{grammar}
\caption{\INI grammar generalized and derived from a single sample.}
\label{fig:ini-grammar}
\end{figure}

From the INI4J parser and one input sample, \EXAGRAM derives the grammar in \autoref{fig:ini-grammar}.  As with Java properties, the \emph{key=value} structure is well identified (\lterm{LINE1} and \lterm{LINE2}), and the grammar is 100\% sound.  However, it is only 16.1\% complete, because the golden grammar also produces underscores for identifiers, which are not found in our sample.  Adding a sample identifier with an underscore would easily fix this.

\subsection{JSON Inputs}

\begin{figure}[t]
\begin{sampleandarrow}\footnotesize
\{
  "glossary": \{
    "title": "example glossary",
        "GlossSeeAlso": 
            ["GML", "XML"] ,
        "bool1"  : true,
        "bool2"  : false,
        "number1" : 2349872,
        "number2" : -45242,
        "number3" : 2349.872,
        "number4" : -98.72,
        "empty" : null,
        "number5" : 2372e71,
        "number6" : 123e-31,
        "number7" : 23.72e71,
        "number8" : 12.83e-33
  \}
\}
\end{sampleandarrow}
\vspace{-.5\baselineskip}
\begin{grammar}
\lterm{MAIN} \expandsinto \rterm{VALUE}
\lterm{VALUE} \expandsinto \rterm{STRING} | \rterm{FALSE} | \rterm{TRUE} |
    \rterm{OBJECT} | \rterm{ARRAY} | \rterm{NULL} | \rterm{NUMBER}
\lterm{STRING} \expandsinto '"' \rterm{HASH} '"'
\lterm{HASH} \expandsinto \pterm{ALPHANUMWITHSPECIALWHITESPACES}
\lterm{FALSE} \expandsinto 'false'
\lterm{TRUE} \expandsinto 'true'
\lterm{OBJECT} \expandsinto '\{' 
    [ \rterm{STRINGINTERNAL} ':' \rterm{VALUE}
     (',' \rterm{STRINGINTERNAL} ':' \rterm{VALUE})* ] '\}'
\lterm{STRINGINTERNAL} \expandsinto '"' \rterm{HASH} '"'
\lterm{ARRAY} \expandsinto '[' [ \rterm{VALUE} ( ',' \rterm{VALUE} )* ] ']'
\lterm{NULL} \expandsinto 'null'
\lterm{NUMBER} \expandsinto \rterm{INTEGER} [[\rterm{FRACTION}] [\rterm{EXPONENT}]]
\lterm{EXPONENT} \expandsinto 'e' ['-'] \pterm{POSITIVEINTEGER}
\lterm{FRACTION} \expandsinto '.' \pterm{POSITIVEINTEGER}
\end{grammar}
\caption{\JSON grammar generalized and derived from a single sample. Whitespace processing omitted.}
\label{fig:json-grammar}
\end{figure}

The most complex input language we have studied with \EXAGRAM is \JSON from the \emph{Minimal \JSON parser}.  Our input sample is set up to cover the \JSON data types, and these are all reflected in the resulting grammar (\autoref{fig:json-grammar}).  Via its membership queries, \EXAGRAM has generalized that objects and arrays can have an arbitrary number of members, and also generalized all number rules.  (Minimal \JSON has its own parser for numeric values.)  This grammar represents all valid \JSON inputs; in our experiments, it is 100\% sound and 100\% complete.

\subsection{Threats to Validity}

All results in the evaluation are subject to threats to validity.  In terms of \emph{external validity,} we do not claim that any of the results or measures shown would or even can generalize to general-purpose programs; rather, our results show the potential of grammar inference even in the absence of large sample sets.  In terms of \emph{internal validity,} it is clear that the completeness of samples easily determines the completeness of the resulting grammars; we have thus set up our samples with a minimum of features.

\section{Conclusion}
\label{sec:conclusion}

With \EXAGRAM, it is possible to derive a complete input grammar given a program and only a minimum of input samples, tracking the flow of input characters and derived values through the program to derive an initial structure, and using active learning to systematically generalize the inferred grammar.  The resulting grammars have a large variety of applications such as reverse engineering of input formats over automatically deriving parsers to decompose inputs into their constituents.  The first and foremost application, however, will be \emph{test generation,} allowing for massive exploration of the input space simply by deriving a producer from the grammar.

Despite its successes, \EXAGRAM is but a milestone on a long path of challenges.  Future topics in learning input languages from programs include:

\begin{description}
\item[Sample-free input learning.]  This challenge is easily posed: Given a program, is it possible to synthesize an input sample that could serve as starting point for \EXAGRAM?  This is a question of \emph{test generation:} Essentially, we are looking for a sample input (or a set thereof) that covers a maximum of different input features, and thus functionality.  Unfortunately, the conditions a valid input has to fulfill are so numerous and complex that current approaches to test generation are easily challenged.

\item[Context-sensitive features.]  Many input formats use context sensitive features like prepending sizes to data blocks or check-sums in binary formats. Since these are usually checked to verify the integrity of an input, it is possible to observe these events during program execution by adding instrumentation to corresponding methods and include them in the trace. This could be used to learn input specifications that are more powerful than context free grammars.

\item[Multi-Layer Grammars.]  Complex input formats come at various \emph{layers,} such as a \emph{lexical} and a \emph{syntactical} layer for parsing programming languages first into tokens and then into trees; or a \emph{transport} layer with checksums around the actual content, which comes in its own format.  As each layer will need its own language description, the challenge will be to identify and separate them, and to come up with a coupled model for both.
\end{description}

If we were to put all these into one grand challenge, then it would be as follows: Given a \emph{compiler binary and no other information, derive an input model that can be used as a sound and precise language reference both for humans and automated testing tools.}  This is a hard challenge, and it may still take many years to solve it; but with \EXAGRAM, we feel we may have gotten a bit closer.

\bibliographystyle{abbrv}
\bibliography{tainting,grammars,ngrams}

\end{document}